\documentclass[10pt]{article}
\usepackage[totalwidth=15.0cm,totalheight=20.0cm]{geometry}
\usepackage{latexsym,amsthm,amsmath,amssymb,url}
\newtheorem{proposition}{Proposition}
\newtheorem{lemma}{Lemma}
\newtheorem{corollary}{Corollary}
\newtheorem{definition}{Definition}
\newtheorem{theorem}{Theorem}
\newtheorem{remark}{Remark}
\newcommand{\LCCM}{{\sc List Coloring Clique Modulator}}
\newcommand{\CM}{{\sc Clique Modulator}}
\newcommand{\LCNminusK}{{\sc List Coloring with $n-k$ colors}}
\newcommand{\PCECM}{{\sc Pre-Coloring Extension Clique Modulator}}

\newcommand{\RLC}{{\sc $(n-k)$-Regular List Coloring}}

\newcommand{\PCESN}{{\sc $(n-|Q|)$-Pre-Coloring Extension}}
\newcommand{\PCESQ}{{\sc $(|Q|-k)$-Pre-Coloring Extension}}
\newcommand{\SKC}{{\sc $(n-k)$-Coloring}}

\newcommand{\SB}{\{\,} \newcommand{\SM}{\;{|}\;}
\newcommand{\SE}{\,\}}

\newcommand{\HS}{{\sc Hitting Set}}
\newcommand{\IS}{{\sc Independent Set}}

\newcommand{\cF}{{\mathcal F}}

\newcommand{\bigoh}{{\mathcal O}}

\usepackage[obeyFinal,colorinlistoftodos,textsize=tiny]{todonotes}

\newenvironment{tightcenter}
 {\parskip=0pt\par\nopagebreak\centering}
 {\par\noindent\ignorespacesafterend}

\usepackage{tikz}
\usetikzlibrary{calc}
\usepackage{xargs}
\usepackage{xifthen}
\usepackage{framed}

\usepackage{ctable}

\newlength{\RoundedBoxWidth}
\newsavebox{\GrayRoundedBox}
\newenvironment{GrayBox}[1]%
   {\setlength{\RoundedBoxWidth}{\textwidth-4.5ex}
    \def\boxheading{#1}
    \begin{lrbox}{\GrayRoundedBox}
       \begin{minipage}{\RoundedBoxWidth}%
   }{%
       \end{minipage}
    \end{lrbox}%
    \begin{tightcenter}%
    \begin{tikzpicture}%
       \node(Text)[draw=black!20,fill=white,rounded corners,%
             inner sep=2ex,text width=\RoundedBoxWidth]%
             {\usebox{\GrayRoundedBox}};
        \coordinate(x) at (current bounding box.north west);
        \node [draw=white,rectangle,inner sep=3pt,anchor=north west,fill=white] 
        at ($(x)+(6pt,.75em)$) {\boxheading};
    \end{tikzpicture}
    \end{tightcenter}\vspace{0pt}%
    \ignorespacesafterend
}    

\newenvironment{problem}[2][]{\noindent\ignorespaces%
                                \FrameSep=6pt%
                                \parindent=0pt%
                \vspace*{-.5em}
                \ifthenelse{\isempty{#1}}{%
                  \begin{GrayBox}{\textsc{#2}}%
                }{%
                  \begin{GrayBox}{\textsc{#2} parameterized by~{#1}}%
                }
                \newcommand\Prob{Problem:}%
                \newcommand\Input{Input:}%
                \begin{tabular*}{\textwidth}{@{\hspace{.1em}} >{\itshape} p{1.2cm} p{0.85\textwidth} @{}}%
            }{
                \end{tabular*}%
                \end{GrayBox}%
                \vspace*{-.5em}
                \ignorespacesafterend
            }

\newtheorem{reduction rule}{\bf Reduction Rule}
\newtheorem{observation}{\bf Observation}

\title{Parameterized Pre-coloring Extension and List Coloring Problems} 

\author{Gregory Gutin, Diptapriyo Majumdar\\Royal Holloway, University of London, UK \and Sebastian Ordyniak\\ University of Sheffield, UK \and Magnus Wahlstr{\"o}m\\Royal Holloway, University of London, UK}

\begin{document}

\maketitle

\begin{abstract}
Golovach, Paulusma and Song (Inf. Comput. 2014) asked to determine the
parameterized complexity of the following problems parameterized by
$k$: (1)  Given a graph $G$, a clique modulator $D$ (a \emph{clique
  modulator} is a set of vertices, whose removal results in a clique)
of size $k$ for $G$, and a
list $L(v)$ of colors for every $v\in V(G)$,
decide whether $G$ has a proper list coloring; (2) Given a graph $G$,
a clique modulator $D$ of size $k$ for $G$, and a pre-coloring $\lambda_P: X \rightarrow Q$ for $X
\subseteq V(G),$   decide whether $\lambda_P$ can be extended to a
proper coloring of $G$ using only colors from $Q.$ For Problem 1 we
design an $\bigoh^*(2^k)$-time randomized algorithm and for Problem 2
we obtain a kernel with at most $3k$ vertices. Banik et al. (IWOCA
2019) proved the the following problem is fixed-parameter tractable
and asked whether it admits a polynomial kernel: Given a graph $G$, an
integer $k$, and a list $L(v)$ of exactly $n-k$ colors for every $v
\in V(G),$ decide whether there is a proper list coloring for $G.$ We
obtain a kernel with $\bigoh(k^2)$ vertices and colors and a
compression to a variation of the problem with $\bigoh(k)$ vertices
and $\bigoh(k^2)$ colors.

\end{abstract}


\section{Introduction}

Graph coloring is a central topic in Computer Science and Graph Theory
due to its importance in theory and applications. Every text book in
Graph Theory has at least a chapter devoted to the topic and
the monograph of Jensen and Toft \cite{JensenT} is completely devoted
to graph coloring problems focusing especially on more than 200
unsolved ones. There are many survey papers on the topic including
recent ones such as~\cite{Chudnovsky,GolovachJPS17,RanderathS,Tuza}.

For a graph $G$, a {\em proper coloring} is a function $\lambda:\
V(G)\rightarrow \mathbb{N}_{\ge 1}$ such that for no pair $u,v$ of
adjacent vertices of $G$, $\lambda(u)=\lambda(v).$
In the widely studied {\sc Coloring} problem, given a graph $G$ and a
positive integer $p$, we are to decide whether there is a proper
coloring $\lambda:\ V(G)\rightarrow [p],$ where
henceforth $[p]=\{1,\dots ,p\}.$   
In this paper, we consider two extensions of {\sc Coloring}: the {\sc
  Pre-Coloring Extension} problem and the {\sc List Coloring}
problem.
In  the {\sc Pre-Coloring Extension} problem, given a graph $G,$ a set $Q$ of colors, and a {\em pre-coloring}
$\lambda_P:\ X \rightarrow Q,$ where $X \subseteq V(G),$ we are to
decide whether there is a proper coloring $\lambda :\ V(G) \rightarrow
Q$ such that $\lambda(x)=\lambda_P(x)$ for every $x\in X.$
In the {\sc List Coloring} problem,  given a graph $G$ and a list $L(u)$ of possible colors for every vertex $u$ of $G$, we are to decide whether $G$ has a proper coloring $\lambda$ such that
$\lambda(u)\in L(u)$ for every vertex $u$ of $G.$ Such a coloring
$\lambda$ is called a {\em proper list coloring}.
Clearly, {\sc Pre-Coloring Extension} is a special case of {\sc List
  Coloring}, where all lists of vertices $x\in X$ are singletons.

The $p$-{\sc Coloring} problem is a special case of  {\sc Coloring}
when $p$ is fixed (i.e., not part of input).  When $Q\subseteq [p]$
($L(u)\subseteq [p]$, respectively), {\sc Pre-Coloring Extension}
({\sc List Coloring}, respectively)
are called  {\sc $p$-Pre-Coloring Extension} ({\sc List $p$-Coloring}, respectively). In classical complexity, it is well-known that $p$-{\sc Coloring},  {\sc $p$-Pre-Coloring Extension} and {\sc List $p$-Coloring} are polynomial-time solvable for $p\le 2$, and the three problems become NP-complete for every $p\ge 3$ \cite{Lovasz,RanderathS}. In this paper, we solve several open problems about pre-coloring extension and list coloring problems, which lie outside classical complexity, so-called parameterized problems. 
We provide basic notions on parameterized complexity in the next section. For more information on parameterized complexity, see recent books \cite{cygan2015,downey2013,fomin2019}.

The first two problems we study are the following ones stated by Golovach et al. \cite{golovachPS14} (see also
\cite{Paulusma15}) who asked to determine their parameterized
complexity. These questions were motivated by a result of Cai
\cite{cai03} who showed that {\sc Coloring Clique Modulator} (the
special case of \PCECM{} when $X=\emptyset$) is fixed-parameter
tractable ({\sf FPT}). Note that a \emph{clique modulator} of a graph
$G$ is a set $D$ of vertices such that $G-D$ is a clique. When using
the size of a clique modulator as a parameter we will for convinience
assume that the modulator is given as part of the input. Note that
this assumption is not necessary (however it avoids having to repeat
how to compute a clique modulator) as we will show in
Section~\ref{ssec:cliquemod} that computing a clique
modulator of size $k$ is FPT and can be approximated to within a
factor of two.

\begin{problem}[$k$]{\LCCM}
  \Input & A graph $G$, a clique modulator $D$ of size $k$ for $G$,
  and a list $L(v)$ of colors for every $v\in V(G).$\\
  \Prob  & Is there a proper list coloring for $G$?
\end{problem}

\begin{problem}[$k$]{\PCECM}
  \Input & A graph $G$, a clique modulator $D$ of size $k$ for $G$, and a pre-coloring
  $\lambda_P: X \rightarrow Q$ for $X \subseteq V(G)$ where $Q$ is a
  set of colors.\\
  
  \Prob  & Can $\lambda_P$ be extended to a proper coloring of $G$
  using only colors from $Q$?
\end{problem}

In Section \ref{sec:cm} we show that {\LCCM} is {\sf FPT}. We first show
a randomized $\bigoh^*(2^{k\log k})$-time algorithm, then we improve
the running time to $\bigoh^*(2^{k})$ using more refined tools and approaches.  
We note that the time $\bigoh^*(2^k)$ matches the best known running
time of $\bigoh^*(2^n)$ for \textsc{Chromatic Number} (where
$n=|V(G)|$)~\cite{BjHuKo09}, while applying to a more powerful
parameter.
It is a long-open problem
whether \textsc{Chromatic Number} can be solved in time $\bigoh(2^{c n})$
for some $c<1$ and Cygan et al.~\cite{CyganDLMNOPSW16} ask whether it
is possible to show that such algorithms are impossible
assuming the Strong Exponential Time Hypothesis (SETH). 

We conclude  Section \ref{sec:cm} by showing  that \LCCM{} does not
admit a polynomial kernel unless NP $\subseteq$ coNP/poly.
The reduction used to prove this result allows us to observe that if \LCCM{} could be solved in
time $\bigoh(2^{c k}n^{\bigoh(1)})$ for some $c<1$, then
the well-known \textsc{Set Cover} problem could be solved in time
$\bigoh(2^{c |U|}|\cF|^{\bigoh(1)}),$ where $U$ and ${\cal F}$ are
universe and family of subsets, respectively. The existence of such an
algorithm is open, and it has been conjectured that no such algorithm is
possible under SETH; see Cygan et al.~\cite{CyganDLMNOPSW16}. Thus, up to the assumption of
this conjecture (called {\em Set Cover Conjecture} \cite{KrauthgamerT19}) 
and SETH, our $\bigoh^*(2^{k})$-time algorithm for \LCCM{} is best possible w.r.t. its
dependency on $k$.

In Section \ref{sec:precol}, we consider \PCECM, which is a subproblem of \LCCM{} and 
prove that \PCECM, unlike  \LCCM, admits a polynomial kernel: a linear kernel with at most $3k$ vertices. 
This kernel builds on a known, but counter-intuitive property of bipartite matchings (see Proposition~\ref{pro:matching}), which was previously used in kernelization by Bodlaender et al.~\cite{BodlaenderJK13}.

In Section \ref{sec:regpk}, we study an open problem stated by Banik et al.~\cite{IWOCApaper}. In a classic result, Chor et al.~\cite{chorFJ04} showed that \textsc{Coloring} has a linear vertex kernel parameterized by $k=n-p$, i.e., if the task is to ``save $k$ colors''.  Arora et al.~\cite{AroraBP018} consider the following as a natural extension to list coloring, and show that it is in XP.  Banik et al.~\cite{IWOCApaper} show that the problem is {\sf FPT}, but leave as an open question whether it admits a polynomial kernel. 

\begin{problem}[$k$]{\RLC}
  \Input & A graph $G$ on $n$ vertices, an integer $k$, and a list $L(v)$ of exactly $n-k$ colors for every $v \in V(G)$.\\
  \Prob  & Is there a proper list coloring for $G$?
\end{problem}

We answer this question in affirmative by giving a kernel with $\bigoh(k^2)$ vertices and colors, as well as a compression to a variation of the problem with $\bigoh(k)$ vertices, encodable in $\bigoh(k^2 \log k)$ bits.  
We note that this compression is asymptotically almost tight, as even \textsc{4-Coloring} does not admit a compression into $\bigoh(n^{2-\varepsilon})$ bits for any $\varepsilon > 0$ unless the polynomial hierarchy collapses~\cite{JansenP17}.

This kernel is more intricate than the above. Via known reduction rules from Banik et al.~\cite{IWOCApaper}, we can compute a clique modulator of at most $2k$ vertices (hence our result for \LCCM{} also solves \RLC{} in $2^{O(k)}$ time). 
However, the usual ``crown rules'' (as in~\cite{chorFJ04} and in Section~\ref{sec:precol})
are not easily applied here, due to complications with the color lists. 
Instead, we are able to show a set of $\bigoh(k)$ vertices whose colorability make up the ``most interesting'' part of the problem, leading to the above-mentioned compression and kernel.

In Section~\ref{sec:save-col}, we consider further natural pre-coloring and list coloring variants 
of the ``saving $k$ colors''
problem of Chor et al.~\cite{chorFJ04}. 
We show that the known
fixed-parameter tractability and linear kernelizability~\cite{chorFJ04} carries over to a natural pre-coloring
generalization but fails for a more general list coloring
variant. Since \RLC{} was originally introduced in~\cite{AroraBP018} as a list
coloring variant of the ``saving $k$ colors'' problem, it is
natural to consider other such variants.

We conclude the paper in Section \ref{sec:conc}, where in particular a number of open questions are discussed. 

\section{Preliminaries}

\subsection{Graphs and Matchings}
We consider finite simple undirected graphs.  For basic
terminology on graphs, we refer to a standard textbook~\cite{Diestel12book}.
For an undirected graph $G=(V,E)$ we denote by $V(G)$
the vertex set of $G$ and by $E(G)$ the edge set of $G$.
For a vertex $v \in V(G)$, we denote by $N_G(v)$ and $N_G[v]$ the \emph{open}
respectively \emph{closed neighborhood} of $v$ in $G$, i.e., $N_G(v):=\SB u
\SM \{u,v\} \in E(G)\SE$ and $N_G[v]:=N_G(v)\cup \{v\}$. We extend
this notion in the natural manner to subsets $V' \subseteq V(G)$, by
setting $N_G(V'):=\bigcup_{v \in V'}N_G(v)$ and $N_G[V']:=\bigcup_{v
  \in V'}N_G[v]$. Moreover, we omit the subscript $G$, if the graph
$G$ can be inferred from the context.
If $V'
\subseteq V(G)$, we denote by $G\setminus V'$ the graph obtained from
$G$ after deleting all vertices in $V'$ together with their adjacent
edges and we denote by $G[V']$ the graph induced by the vertices in
$V'$, i.e., $G[V']=G \setminus (V(G)\setminus V')$. 
We say that $G$ is {\em bipartite with bipartition} $(A,B)$, if $\{A,B\}$
partitions $V(G)$ and $G[A]$ as well as $G[B]$ have no edges.

A \emph{matching} $M$ is a subset of $E(G)$ such that no two edges in
$M$ share a common endpoint. We say that $M$ is \emph{maximal} if
there is no edge $e\in E(G)$ such that $M\cup \{e\}$ is a matching
and we say that $M$ is \emph{maximum} if it is maximal and there is no
maximal matching in $G$ containing more edges than $M$. We denote by
$V(M)$ the set of all endpoints of the edges in $M$, i.e.,
the set $\bigcup_{e\in M}e$. We say that
$M$ {\em saturates} a subset $V'\subset V(G)$ if $V' \subseteq V(M)$.

Let $H=(V,E)$ be an undirected bipartite graph with bipartition $(A,B)$.
We say that a set $C$ is a
\emph{Hall set} for $A$ or $B$ if $C \subseteq A$ or $C \subseteq B$,
respectively, and $|N_H(C)|<|C|$. 

We will need the following well-known properties for matchings.
\begin{proposition}[Hall's Theorem~\cite{Diestel12book}]\label{pro:hallset}
  Let $G$ be an undirected bipartite graph with bipartition
  $(A,B)$. Then $G$ has a matching saturating $A$ if and only if there
  is no Hall set for $A$, i.e., for
  every $A' \subseteq A$, it holds that $|N(A')|\geq |A'|$. 
\end{proposition}

\begin{proposition}[{\cite[Theorem 2]{BodlaenderJK13}}]\label{pro:matching}
  Let $G$ be a bipartite graph with bipartition $(X,Y)$ and let $X_M$
  be the set of all vertices in $X$ that are endpoints of a maximum
  matching $M$ of $G$. Then, for every $Y' \subseteq Y$, it holds that
  $G$ contains a matching that covers $Y'$ if and only if so
  does $G[X_M\cup Y]$.
\end{proposition}

\subsection{Parameterized Complexity}

An instance of a parameterized problem $\Pi$
is a pair $(I,k)$ where $I$ is the \emph{main part} and $k$ is the
\emph{parameter}; the latter is usually a non-negative integer.  
A parameterized problem is
\emph{fixed-parameter tractable} ({\sf FPT}) if there exists a computable function
$f$ such that instances $(I,k)$ can be solved in time $\bigoh(f(k)|{I}|^c)$
where $|I|$ denotes the size of~$I$ and $c$ is an absolute constant. The class of all fixed-parameter
tractable decision problems is called {\sf FPT} and algorithms which run in
the time specified above are called {\sf FPT} algorithms. As in other literature on {\sf FPT} algorithms,
we will often omit the polynomial factor in $\bigoh(f(k)|{I}|^c)$ and write $\bigoh^*(f(k))$ instead.
To establish that a problem under a specific parameterization is not
in {\sf FPT} 
we prove that
it is W[1]-hard as it is widely believed that {\sf FPT}$\neq$W[1]. 

A {\em reduction rule} $R$ for a parameterized problem $\Pi$  is an algorithm $A$ that
given an instance $(I,k)$ of a problem $\Pi$ returns an instance $(I',k')$ of
the \emph{same} problem. The reduction rule is said to be {\em safe} if it
holds that  $(I,k) \in \Pi$ if and only if $(I',k') \in \Pi$. If $A$ runs in polynomial time in $|I|+k$ then
$R$ is a {\em polynomial-time reduction rule}. Often we omit the adjectives ``safe'' and ``polynomial-time''
in ``safe polynomial-time reduction rule'' as we consider only such reduction rules.


A \emph{kernelization} (or, a {\em kernel}) of a parameterized problem $\Pi$ is a reduction rule
such that $|I'|+k'\leq f(k)$ for some computable function
$f$. It is not hard to show that a decidable parameterized problem is {\sf FPT} 
if and only if it admits a kernel \cite{cygan2015,downey2013,fomin2019}.
The function $f$ is called the {\em size} of the kernel, and we
have a {\em polynomial kernel} if $f(k)$ is polynomially bounded in $k$.  

A kernelization can be generalized by considering a reduction (rule) from a 
parameterized problem $\Pi$ to another parameterized problem $\Pi'.$
Then instead of a kernel we obtain a {\em generalized kernel} (also called a bikernel \cite{AlonGKSY11} in the literature).
If the problem $\Pi'$ is not parameterized, then a reduction from $\Pi$ to $\Pi'$ (i.e., $(I,k)$ to $I'$) 
is called a {\em compression}, which is {\em polynomial} if $|I'|\le p(k),$ where $p$ is 
a fixed polynomial in $k.$ If there is a polynomial compression from $\Pi$ to $\Pi'$ and $\Pi'$ is 
polynomial-time reducible back to $\Pi$, then combining the compression with the reduction
gives a polynomial kernel for $\Pi.$

\subsection{Clique Modulator}\label{ssec:cliquemod}

Let $G$ be an undirected graph. We say that a set $D \subseteq V(G)$
is a \emph{clique modulator} for $G$ if $G- D$ is a
clique. Since we will use the size of a smallest clique modulator as a
parameter for our coloring problems, it is natural to ask whether
the following problem can be solved efficiently.

\begin{problem}[$k$]{\CM}
  \Input & A graph $G$ and an integer $k$\\
  \Prob  & Does $G$ have a clique modulator of size at most $k$?
\end{problem}

The following proposition shows that this is indeed the case. Namely,
\CM{} is both {\sf FPT} and can be approximated
within a factor of two. The former is important for our
{\sf FPT} algorithms and the later for our kernelization
algorithms as it allows us to not depend on a clique modulator given
as part of the input.
\begin{proposition}
  \CM{} is fixed-parameter tractable (in time
  $\bigoh^*(1.2738^k)$) and can be approximated within a
  factor of two.
\end{proposition}
\begin{proof}
  It is straightforward to verify that a graph $G$ has a clique
  modulator of size at most $k$ if and only if the complement
  $\overline{G}$ of $G$ has a vertex cover of size at most
  $k$. The statement now follows from the fact that the vertex cover
  problem is fixed-parameter tractable~\cite{ChenKX10} (in time
  $\bigoh^*(1.2738^k)$) and can be approximated within
  a factor of two~\cite{GareyJ79}.
\end{proof}

\section{List Coloring Clique Modulator}
\label{sec:cm}

The following lemma is often used in the design of randomized algorithms. 

\begin{lemma}\label{lem:SZ} (Schwartz-Zippel \cite{Schwartz1980,Zippel1979}). Let $P(x_1, \ldots , x_n)$ be a multivariate polynomial of total
degree at most $d$ over a field $\mathbb{F}$, and assume that $P$ is not identically zero. Pick $r_1,\dots, r_n$
uniformly at random from $\mathbb{F}$. Then Pr$[P(r_1,\ldots ,r_n)=0]\le d/|\mathbb{F}|.$
\end{lemma}

Both parts of the next lemma will be used in this section. The part for fields of characteristic two was proved by Wahlstr{\"o}m \cite{Wahlstrom13}. The part for reals can be proved similarly.

\begin{lemma}\label{lem:coef}
Let $P(x_1, \ldots , x_n)$ be a polynomial over a field of characteristic two (over reals, respectively), and $J \subseteq [n]$
a set of indices. For a set $I \subseteq [n],$ define $P_{- I}(x_1,\ldots  , x_n) = P(y_1,\ldots  , y_n),$ where $y_i = 0$
for $i \in I$ and $y_i = x_i$, otherwise. Define
\begin{eqnarray*}
Q(x_1,\ldots , x_n) & = & \sum_{I\subseteq J} P_{- I}(x_1,\ldots , x_n)\\
(Q(x_1,\ldots , x_n) & = & \sum_{I\subseteq J} (-1)^{|I|} P_{- I}(x_1,\ldots , x_n), \mbox{ respectively}).
\end{eqnarray*}
Then for any monomial $T$ divisible by $\Pi_{i\in J} x_i$ we have ${\rm coef}_Q T = {\rm coef}_P T,$ and for
every other monomial $T$ we have ${\rm coef}_Q T = 0.$
\end{lemma}

Now we are ready to prove the first new result of this section.

\begin{theorem}\label{theorem:slow-fpt}
\LCCM {} can be solved by a randomized algorithm in time $\bigoh^*(2^{k\log k}).$
\end{theorem}
\begin{proof}
  Let $L=\bigcup_{V \in V(G)}L(v)$ and $C=G-D$. We say that a proper list coloring
  $\lambda$ for $G$ is compatible with $(\mathcal{D},\mathcal{D}')$
  if:
  \begin{itemize}
  \item $\mathcal{D}=\{D_1,\dotsc,D_p\}$ is the partition of all
    vertices in $D$ that do not reuse colors used by $\lambda$ in $C$
    into color classes given by $\lambda$ and
  \item $\mathcal{D}=\{D_1',\dotsc,D_t'\}$ is the partition of all
    vertices in $D$ that do reuse colors used by $\lambda$ in $C$
    into color classes given by $\lambda$.
  \end{itemize}
  Note that $\{D_1,\dotsc,D_p,D_1',\dotsc,D_t'\}$ is the partition of
  $D$ into colorclasses given by $\lambda$.
  
  For a given pair $(\mathcal{D},\mathcal{D}')$, we will now construct
  a bipartite graph $B$ (with weights on its
  edges) such that $B$ has a perfect matching satisfying certain
  additional properties if and only if $G$ has a proper list coloring
  that is compatible with $(\mathcal{D},\mathcal{D}')$. 
  $B$ has bipartition $(C \cup \{D_1,\dotsc,D_p\},L)$ and an edge
  between a vertex $c\in C$ and a vertex $\ell
  \in L$ if and only if $\ell \in L(u)$. Moreover, $B$ has an edge
  between a vertex $D_i$ and a
  vertex $\ell \in L$ if and only if $\ell \in \bigcap_{d \in
    D_i}L(d)$. Finally,  if
  $c\in C$ and $\ell \in L$, then assign the edge $c\ell$ weight $\sum_{j\in J} x_j$,
  where $x_j$'s are variables and $j\in J$ if and only if 
  $\ell \in (\bigcap_{d \in D_j'}L(d))\cap L(c)$
  and $c$ is not adjacent to any vertex in $D_j'$. All other edges in
  $B$ are given weight $1$. In the following we will assume that $B$
  is balanced; if this is not the case then we simply add the right
  amount of dummy vertices to the smaller side and make them adjancent
  (with an edge of weight 1) to all vertices in the opposite
  side. Note that $B$ has a perfect matching $M$ such that there is a 
  bijection $\alpha$ between $[t]$ and $t$ edges in $M$ such that for every
  $i \in [t]$, the weight of the edge $\alpha(i)$ contains the term
  $x_i$ if and only if $G$ has a proper list coloring that is
  compatible with $(\mathcal{D},\mathcal{D}')$.

  Let $M$ be the weighted incidence matrix of $B$, i.e., $M$ is an
  $|V(B)/2|\times |V(B)/2|$ matrix such that its entries $L_{i,j}$ 
  equal to the weight of the edge between the $i$-th vertex on one
  side and the $j$-th vertex on the other side of $B$ if it exists and $L_{i,j}=0$
  otherwise.

  Note that the permanent per$(M)$ of $M$ equals to the sum of the
  products of entries of $M$, where each product corresponds to a
  perfect matching $Q$ of $B$ and is equal to the product of the entries
  of $M$ corresponding to the edges of $Q$. Some of the entries of $M$
  contain sums of variables $x_j$, $j\in [t]$ and thus per$(M)$  is a
  polynomial in these variables.
  
  Now it is not hard to see that per$(M)$ contains the monomial
  $\prod_{j=1}^tx_j$ if and only if $B$ has a perfect matching $M$ such that there is a 
  bijection $\alpha$ between $[t]$ and $t$ edges in $M$ such that for every
  $i \in [t]$, the weight of the edge $\alpha(i)$ contains the term
  $x_i$, which in turn is equivalent to $G$ having a proper list coloring that is
  compatible with $(\mathcal{D},\mathcal{D}')$.

  Hence, deciding whether $G$ has a proper list coloring that is
  compatible with $(\mathcal{D},\mathcal{D}')$ boils down to deciding
  whether the permant of $M$ contains the monomial $\prod_{j=1}^t x_j$.
  For any evaluation of
  variables $x_j$, we can compute per$(M)$ over the field of
  characteristic two by replacing permanent with determinant, which
  can be computed in polynomial-time~\cite{BunchHopcroft74}. 

  Now let $P(x_1,\dots ,x_t)={\rm det}(M)$ and $Q(x_1,\ldots , x_t)
  =\sum_{I\subseteq [t]} P_{- I}(x_1,\ldots , x_t).$
  Note that $Q(x_1,\dotsc,x_t)\neq 0$ if and only if det$(M)$ contains the
  monomial $\prod_{j=1}^t x_j$. Moreover, using Lemmas
  \ref{lem:SZ} and \ref{lem:coef} (with $P$ and $Q$ just defined), we
  can verify in time $\bigoh^*(2^t)$ whether $Q(x_1,\dotsc,x_t)=0$ (i.e. whether
  det$(M)$ contains the monomial $\prod_{j=1}^t x_j$) with probability
  at least $1-\frac{t}{|\mathbb{F}|}\ge 1-\frac{1}{t}$ for a field
  $\mathbb{F}$ of characteristic 2 such that $|\mathbb{F}|\ge t^2$.

  Our algorithm sets $t=k$ and for every pair
  $(\mathcal{D},\mathcal{D}')$, where $\mathcal{D}\cup\mathcal{D}'$ is
  a partition of $D$ into independent sets, constructs graph $B$ and matrix $M$. It then
  verifies in time $\bigoh^*(2^t)$ whether $Q(x_1,\dotsc,x_t)=0$ and if $Q(x_1,\dotsc,x_t)\ne
  0$ it returns `Yes'
  and terminates. If the algorithm runs to the end, it returns `No'.

  Note that the time complexity of the algorithm is dominated by the
  number of choices for $(\mathcal{D},\mathcal{D}')$, which is in turn
  dominated by
  $\bigoh^*({\cal B}_k),$ where ${\cal B}_k$ is the $k$-th Bell
  number. By Berend and Tassa \cite{BerendT10}, ${\cal B}_k<(\frac{0.792k}{\ln
    (k+1)})^k,$ and thus $\bigoh^*({\cal B}_k)=\bigoh^*(2^{k\log
    k}).$

  \newcommand{\old}[1]{}
  \old{
  Let $L=\cup_{v\in V(G)} L(v).$  Suppose there is a proper list coloring $\lambda$ of $G.$ 
  Let $D=G-V(C)$, $D'$  the vertices of $D$ which use the same colors as
  some vertices of $C$ in coloring $\lambda$, and $C'=D\setminus D'.$ Consider
  a partition $(D'_1,\dots ,D'_t)$ of  $D'$ into non-empty subsets each
  of which includes all vertices of $D'$ of the same color in $\lambda$ and
  let $L_i=\cap_{v\in D'_i}L(v)$ for all $i\in [t].$ Construct a
  bipartite graph $B$ with bipartition $(C\cup C',L)$ such that
  $u\ell\in E(B)$, where $u\in C\cup C'$ and $\ell\in L$ if and only if
  $\ell\in L(u).$ If $u\in C'$ we assign the edge $u\ell$ weight 1.  If
  $u\in C$ then assign the edge $u\ell$ weight $\sum_{j\in J} x_j$,
  where $x_j$'s are variables and $j\in J$ if and only if 
  $\ell \in L_j\cap L(u)$
  and $u$ is not adjacent to any vertex in $D_j'$. If 
  $\ell\not\in L_j\cap L(u)$ or $u$ is adjacent to some vertex in $D_j'$,
  then assign weight 1 to $u\ell.$

  We ensure that $B$ is balanced, by adding dummy vertices to the part
  containing $C\cup C'$ or the part containing $L$ and making those
  dummy vertices adjacent (via an edge with weight 1) to all vertices
  of the opposite side.
  
  Note that it is now easy to see that $G$ has a proper list coloring
  Let $M$ be an $|L|\times
  |L|$-matrix corresponding to $B$ such that its entries $L_{i,j}$ are
  equal to the weight of the edge between the $i$-th vertex on one
  side and the $j$-th vertex on the other side of $B$ if it exists and $L_{i,j}=0,$
  otherwise.

  Note that the permanent per$(M)$ of $M$ equals to the sum of the
  products of entries of $M$, where each product corresponds to a
  perfect matching $Q$ of $B$ and equal to the product of the entries
  of $M$ corresponding to the edges of $Q.$ Some of the entries of $M$
  contain sums of variables $x_j$, $j\in [t]$ and thus per$(M)$  is a
  polynomial in these variables.
  
  By the definitions of per$(M)$ and the partition $(D'_1,\dots ,D'_t),$ as a polynomial per$(M)$ must contain a monomial $\prod_{j=1}^t x_j.$ Now it is not hard to see that per$(M)$ contains the monomial $\prod_{j=1}^t x_j$ if and only if $G$ has a proper list coloring with $n-t$ colors. For any evaluation of variables $x_j$, we can compute per$(M)$ over the field of characteristic two by replacing permanent with determinant. Then the computation can be done in polynomial time. 

  Now let $P(x_1,\dots ,x_t)={\rm det}(M)$ and $Q(x_1,\ldots , x_t) =\sum_{I\subseteq [t]} P_{- I}(x_1,\ldots , x_t).$ Using Lemmas \ref{lem:SZ} and \ref{lem:coef} (with $P$ and $Q$ just defined), we can verify in time $\bigoh^*(2^t)$ whether det$(M)=0$ contains the monomial $\prod_{j=1}^t x_j$ with probability at least $1-\frac{t}{|\mathbb{F}|}\ge 1-\frac{1}{t}$ for a field $\mathbb{F}$ of characteristic 2 such that $|\mathbb{F}|\ge t^2.$

  Our algorithm sets $t=k$ and for every subset  $D'$ of $D$ and every partition $(D'_1,\dots ,D'_t)$ of $D',$ where all parts $D_i$ are independent sets, constructs graph $B$ and matrix $M$. It then verifies  whether det$(M)=0$ and if  det$(M)\ne 0$ it returns `Yes' and terminates. If the algorithm runs to the end, it returns `No'. 

  Note that the time complexity of the algorithm is dominated by the number of choices for $D'$ and the partitions of $D'$, which is $\bigoh^*(2^k{\cal B}_k),$ where ${\cal B}_k$ is the $k$th Bell number.  
  By Berend and Tassa \cite{BerendT10}, ${\cal B}_k<(\frac{0.792k}{\ln (k+1)})^k,$ and thus 
  $\bigoh^*(2^k{\cal B}_k)=\bigoh^*(2^{k\log k}).$
}
\end{proof}

\subsection{A faster FPT algorithm}

We now show a faster {\sf FPT} algorithm, running in
time~$\bigoh^*(2^k)$. It is a variation on the same algebraic sieving
technique as above, but instead of guessing a partition of the
modulator it works over a more complex matrix. We begin by defining
the matrix, then we show how to perform the sieving step in
$\bigoh^*(2^k)$ time.

\subsubsection{Matrix definition} 

As before, let $L=\bigcup_{v \in V(G)} L(v)$ be the set of all
colours, and let $C=G-D$.  Define an
auxiliary bipartite graph $H=(U_H \cup V_H, E_H)$ where initially
$U_H=V(G)$ and $V_H=L$, and where $v\ell \in E_H$ for $v \in V(G)$,
$\ell \in L$ if and only if $\ell \in L(v)$.  Additionally, introduce
a set $L'=\{ \ell_d' \SM d \in D\}$ of $k$ artificial colours, add
$L'$ to $V_H$, and for each $d\in D$ connect $\ell_d'$ to $d$ but
to no other vertex.  Finally, pad $U_H$ with $|V_H|-|U_H|$ artificial
vertices connected to all of $V_H$; note that this is a nonnegative
number, since otherwise $|L|<|V(C)|$ and we may reject the instance.

Next, we associate with every edge $v\ell \in E_H$ a set $S(v\ell) \subseteq 2^D$ as follows.
\begin{itemize}
\item If $v \in V(C)$, then $S(v\ell)$ contains all sets $S \subseteq D$ such that the following hold:
  \begin{enumerate}
  \item $S$ is an independent set in $G$
  \item $N(v) \cap S = \emptyset$
  \item $\ell \in \bigcap_{s \in S} L(s)$.
  \end{enumerate}
\item If $v \in D$ and $\ell \in L$, then $S(v\ell)$ contains all sets $S \subseteq D$ such that the following hold:
  \begin{enumerate}
  \item $v \in S$
  \item $S$ is an independent set in $G$
  \item $\ell \in \bigcap_{s \in S} L(s)$.
  \end{enumerate}
\item If $v$ or $\ell$ is an artificial vertex -- in particular, if
  $\ell=\ell_d'$ for some $d \in D$ -- then $S(v\ell)=\{\emptyset\}$.
\end{itemize}
Finally, define a matrix $A$ of dimensions $|U_H| \times |V_H|$, with
rows labelled by $U_H$ and columns labelled by $V_H$, whose entries
are polynomials as follows. Define a set of variables $X=\{x_d\SM d\in
D\}$ corresponding to vertices of $D$, and additionally a set $Y=\{y_e \mid e \in E_H\}$. Then for every edge $v\ell$ in $H$, $v \in U_H$, $\ell \in V_H$ we define
\[
P(v\ell)=\sum_{S \in S(v\ell)} \prod_{s \in S} x_s,
\]
where as usual an empty product equals $1$. Then for each edge $v\ell \in E_H$ we let $A[v,\ell] = y_{v\ell} P(v\ell)$, and the remaining entries of $A$ are $0$. 
We argue the following. (Expert readers may note although the argument can be sharpened to show the existence of a multilinear term, we do not wish to argue that there exists such a term with odd coefficient. Therefore we use the simpler sieving of Lemma~\ref{lem:coef} instead of full multilinear detection, cf.~\cite{cygan2015}.)

\begin{lemma} 
  \label{lemma:matrix-solution}
  Let $A$ be defined as above. Then $\det A$ (as a polynomial) contains a monomial divisible by $\prod_{x \in X} x$ if and only if $G$ is properly list colorable. 
\end{lemma}
\begin{proof}
  We first note that no cancellation happens in $\det A$.  Note that monomials of $\det A$ correspond (many-to-one) to perfect matchings of $H$, and thanks to the formal variables $Y$, two monomials corresponding to distinct perfect matchings never interact. On the other hand, if we fix a perfect matching $M$ in $H$, then the contributions of $M$ to $\det A$ equal $\sigma_M \prod_{e \in M} y_e P(e)$, where $\sigma_M \in \{1,-1\}$ is a sign term depending only on $M$. Since the polynomials $P(e)$ contain only positive coefficients, no cancellation occur, and every selection of a perfect matching $M$ of $H$ and a factor from every polynomial $P(e)$, $e \in M$ results (many-to-one) to a monomial with non-zero coefficient in $\det A$.

  We now proceed with the proof.  On the one hand, let $c$ be a proper list coloring of $G$. Define an ordering $\prec$ on $V(G)$ such that $V(C)$ precedes $D$, and define a matching $M$ as follows. For every vertex $v \in V(C)$, add $vc(v)$ to $M$. For every vertex $v \in D$, add $vc(v)$ to $M$ if $v$ is the first vertex according to $\prec$ that uses color $c(v)$, otherwise add $v\ell_v'$ to $M$. Note that $M$ is a matching in $H$ of $|V(G)|$ edges.  Pad $M$ to a perfect matching in $H$ by adding arbitrary edges connected to the artificial vertices in $U_H$; note that this is always possible. Finally, for every edge $v\ell \in M$ with $\ell \in L$ we let $D_{v\ell}=D \cap c^{-1}(\ell)$.  Observe that for every edge $v\ell$ in $M$, $D_{v\ell} \in S(v\ell)$; indeed, this holds by construction of $S(v\ell)$ and since $c$ is a proper list coloring.  Further let $p_{v\ell}=\prod_{v \in D_{v\ell}} x_v$; thus $p_{v\ell}$ is a term of $P(v\ell)$.  It follows, by the discussion in the first paragraph of the proof, that
\[
\alpha \sigma_M \prod_{v\ell \in M} y_{v\ell} p_{v\ell}
\]
is a monomial of $\det A$ for some constant $\alpha > 0$, where $\sigma_M \in \{1, -1\}$ is the sign term for $M$.  It remains to verify that every variable $x_d \in X$ occurs in some term $p_{v\ell}$. Let $\ell=c(d)$ and let $v$ be the earliest vertex according to $\prec$ such that $c(v)=\ell$.  Then $v\ell \in M$ and $x_d$ occurs in $p_{v\ell}$.  This finishes the first direction of the proof.

On the other hand, assume that $\det A$ contains a monomial $T$ divisible by $\prod_{x \in X} x$, and let $M$ be the corresponding perfect matching of $H$.  Let $T=\alpha \prod_{e \in M} y_e p_e$ for some constant factor $\alpha$, where $p_e$ is a term of $P(e)$ for every $e \in M$.  Clearly such a selection is possible; if it is ambiguous, make the selection arbitrarily.  Now define a mapping $c \colon V(G) \to L$ as follows.  For $v \in V(C)$, let $v\ell \in M$ be the unique edge connected to $v$, and set $c(v)=\ell$.  For $v \in D$, let $v'$ be the earliest vertex according to $\prec$ such that $x_v$ occurs in $p_{v'\ell}$, where $v'\ell \in M$. Set $c(v)=\ell$.  We verify that $c$ is a proper list coloring of $G$. First of all, note that $c(v)$ is defined for every $v \in V(G)$ and that $c(v) \in L(v)$.  Indeed, if $v \in V(C)$ then $c(v) \in L(v)$ since $vc(v) \in E_H$; and if $v \in D$ then $c(v) \in L(v)$ is verified in the creation of the term $p_{vc(v)}$ in $P(vc(v))$.  Next, consider two vertices $u, v \in V(G)$ with $c(u)=c(v)$.  If $u, v \in D$, then $u$ and $v$ are represented in the same term $p_{v'c(v)}$ for some $v'$, hence $u$ and $v$ form an independent set; otherwise assume $u \in V(C)$.  Note that $u, v \in V(C)$ is impossible since otherwise the matching $M$ would contain two edges $uc(u)$ and $vc(u)$ which intersect.  Thus $v \in D$, and $v$ is represented in the term $p_{uc(u)}$.  Therefore $uv \notin E(G)$, by construction of $P(uc(u))$. 
We conclude that $c$ is a proper coloring respecting the lists $L(v)$, i.e., a proper list coloring. 
\end{proof}

\subsubsection{Fast evaluation} 

By the above description, we can test for the existence of a list coloring of $G$ using $2^k$ evaluations of $\det A$, as in Theorem~\ref{theorem:slow-fpt}; and each evaluation can be performed in $\bigoh^*(2^k)$ time, including the time to evaluate the polynomials $P(v\ell)$, making for a running time of $\bigoh^*(4^k)$ in total (or $\bigoh^*(3^k)$ with more careful analysis).  We show how to perform the entire sieving in time $\bigoh^*(2^k)$ using fast subset convolution.  

For $I \subseteq D$, let us define $A_{-I}$ as $A$ with all occurrences of variables $x_i$, $i \in I$ replaced by 0, and for every edge $v\ell$ of $H$, let $P(v\ell)_{-I}$ denote the polynomial $P(v\ell)$ with $x_i$, $i \in I$ replaced by 0. Then a generic entry $(v,\ell)$ of $A_{-I}$ equals
\[
A_{-I}[v,\ell] = y_{v\ell} P_{-I}(v\ell),
\]
and in order to construct $A_{-I}$ it suffices to precompute the value of $P_{-I}(v\ell)$ for every edge $v\ell \in E_H$, $I \subseteq D$.  For this, we need the  \emph{fast zeta transform} of Yates~\cite{Yates37}, which was introduced to exact algorithms by Bj\"orklund et al.~\cite{BjHuKo09}.

\begin{lemma}[\cite{Yates37,BjHuKo09}] \label{lemma:fastzeta}
Given a function $f: 2^N \rightarrow R$ for some ground set $N$ and ring $R$,
we may compute all values of $\hat f: 2^N \rightarrow R$ defined as
$\hat f(S) = \sum_{A \subseteq S} f(A)$ using $\bigoh^*(2^{|N|})$ ring operations. 
\end{lemma}

We show the following lemma, which is likely to have analogs in the literature, but we provide a short proof for the sake of completeness. 

\begin{lemma}
  \label{lemma:zeta-precompute}
  Given an evaluation of the variables $X$, the value of $P_{-I}(v\ell)$ can be computed for all $I \subseteq D$ and all $v\ell \in E_H$ in time and space $\bigoh^*(2^k)$.
\end{lemma}
\begin{proof}
  Consider an arbitrary polynomial $P_{-I}(v\ell)$.  
  Recalling $P(v\ell)=\sum_{S \in S(v\ell)} \prod_{s \in S} x_s$, we have
\[
P_{-I}(v\ell) = \sum_{S \in S(v\ell)} [S \cap I = \emptyset] \prod_{s \in S} x_s =
 \sum_{S \subseteq (D-I)} [S \in S(v\ell)] \prod_{s \in S} x_s,
\]
using Iverson bracket notation.\footnote{Recall that for a logical proposition $P$, $[P]=1$ if $P$ is true and 0, otherwise.}  Using $f(S)=[S \in S(v\ell)]\prod_{s \in S} x_s$, this clearly fits the form of Lemma~\ref{lemma:fastzeta}, with $\hat f(D-I)=P_{-I}(v\ell)$. Hence we apply Lemma~\ref{lemma:fastzeta} for every edge $v\ell \in E_H$, for $\bigoh^*(2^k)$ time per edge, making $\bigoh^*(2^k)$ time in total to compute all values.
\end{proof}

Having access to these values, it is now easy to complete the algorithm.

\begin{theorem}
  \label{theorem:cliquemod-fast}
  \LCCM{} can be solved by a randomized algorithm in time $\bigoh^*(2^k)$.
\end{theorem}
\begin{proof}
  Let $A$ be the matrix defined above (but do not explicitly construct it yet).  By Lemma~\ref{lemma:matrix-solution}, we need to check whether $\det A$ contains a monomial divisable by $\prod_{x \in X} x$, and by Lemma~\ref{lem:coef} this is equivalent to testing whether
\[
\sum_{I \subseteq D} (-1)^{|I|} \det A_{-I} \not \equiv 0.
\]
By the Schwartz-Zippel lemma, it suffices to randomly evaluate the variables $X$ and $Y$ occurring in $A$ and evaluate this sum once; if $G$ has a proper list coloring and if the values of $X$ and $Y$ are chosen among sufficiently many values, then with high probability the result is non-zero, and if not, then the result is guaranteed to be zero.  Thus the algorithm is as follows.
\begin{enumerate}
\item Instantiate variables of $X$ and $Y$ uniformly at random from $[N]$ for some sufficiently large $N$. Note that for an error probability of $\varepsilon > 0$, it suffices to use $N=\Omega(n^2(1/\varepsilon))$.
\item Use Lemma~\ref{lemma:zeta-precompute} to fill in a table with the value of $P_{-I}(v\ell)$ for all $I$ and $v\ell$ in time $\bigoh^*(2^k).$
\item Compute 
  \[
  \sum_{I \subseteq D} (-1)^{|I|} \det A_{-I},
  \]
  constructing $A_{-I}$ from the values $P_{-I}(v\ell)$ in polynomial time in each step. 
\item Answer YES if the result is non-zero, NO otherwise.
\end{enumerate}
Clearly this runs in total time and space $\bigoh^*(2^k)$ and the correctness follows from the arguments above.
\end{proof}

\subsection{Refuting Polynomial Kernel}

In this section, we prove that {\LCCM} does not admit a polynomial kernel.
We prove this result by a polynomial parameter transformation from {\HS} where the parameter is the number of sets.
Notice that {\HS} parameterized by number of sets is equivalent to {\sc Set Cover} parameterized by the universe size.
Let us recall the formal definition of the {\HS} problem.

\begin{problem}[$m$]{\HS}
  \Input & A universe $U$ of $n$ elements, a family $\cF \subseteq 2^U$ of $m$ subsets of $U$, and an integer $k$.\\
  \Prob  & Is there $X \subseteq U$ with at most $k$ elements such that for every $F \in \cF$, it holds that $F \cap X \neq \emptyset$?
\end{problem}

Let $(U, \cF, k)$ be an instance of {\HS} problem where $U = [n]$, and $\cF = \{F_1,\ldots,F_m\}$.
Now, we are ready to describe the construction.

{\bf Construction:}
For every $i \in [m]$, we create a vertex $u_i$ and assign $L(u_i) = F_i$.
Let $D=\{u_1,\ldots,u_m\}$.
In addition, we create a clique $C$ with  $n - k$ vertices $\{v_1,\ldots,v_{n-k}\}$.
Moreover, for every $j \in [n-k]$, we set $L(v_j) = U$ and
for all $i \in [n-k]$ and $j \in [m]$, let $(u_i,v_j)$ be an edge.
This completes the construction, which takes polynomial time.
We denote the obtained graph by $G.$
Now, we prove the following lemma.

\begin{lemma}
\label{lemma:HS-eqv-List-Coloring}
$(U, \cF, k)$ is a yes-instance if and only if $G$ admits a proper list coloring.
\end{lemma}

\begin{proof}
Towards showing the forward direction, let $(U,\cF, k)$ be an yes-instance.
Then, there is a set $X$ of at most $k$ elements from $U$ such that
for every $F_i \in \cF$, $X \cap F_i \neq \emptyset$.

Using the elements present in $X$, we can color $D$ as follows.
We pick an element arbitrarily from every $F_i \cap X$, and color the
vertex $u_i$ using that color.
After that, we provide different colors to different vertices in $C$
that are different from the colors used in $D$ as well.
Hence, we can color $G$ by $n$ colors.

Towards showing the backwards direction, suppose that $G$ has a proper
list coloring.
Note that all vertices of $C$ have to get different colors.
Hence, the vertices of $D$ must be colorable using only $k$ colors. 
Suppose that $X$ is the set of $k$ colors used to color the vertices of $D$.
Note that the colors respect the list for every vertex in $D$ where the list represents the sets in the family.
Hence, $X$ is a hitting set of size $k$.
\end{proof}

It is known due to Dom et al.~\cite{DomLS14} that {\HS} parameterized by $m$
does not have a polynomial kernel unless NP $\subseteq$ coNP/poly.
Using this fact and Lemma~\ref{lemma:HS-eqv-List-Coloring}, we have the following theorem.

\begin{theorem}
{\LCCM} parameterized by $k$ does not admit a polynomial kernel unless NP $\subseteq$ coNP/poly.
\end{theorem}

Note that the reduction also shows that if \LCCM{} could be solved in
time $\bigoh(2^{\epsilon k}n^{\bigoh(1)})$ for some $\epsilon<1$, then
\HS{} could be solved in time
$\bigoh(2^{\epsilon |\cF|}|U|^{\bigoh(1)})$, which in turn would imply
that any instance $I$ with universe $U$ and set family $\cF$
of the well-known \textsc{Set Cover} problem could be solved in time
$\bigoh(2^{\epsilon |U|}|\cF|^{\bigoh(1)})$. The existence of such an algorithm
is open, and it has been conjectured that no such algorithm is
possible under SETH (the strong exponential-time hypothesis); see
Cygan et al.~\cite{CyganDLMNOPSW16}. Thus, up to the assumption of
this conjecture and SETH, the algorithm for \LCCM{} given in
Theorem~\ref{theorem:cliquemod-fast} is best possible w.r.t. its
dependency on $k$.

\section{Polynomial kernel for {\PCECM}}\label{sec:precol}

In the following let $(G,D,k,\lambda_P, X, Q)$ be an instance of
\PCECM{}, let $C=G-D$, let $D_P$ be the set of all precolored vertices
in $D$, and let $D'=D \setminus D_P$.

\begin{reduction rule}
  \label{red-rule:large-non-neighbor}
  Remove any vertex $v \in D'$ that has less than $|Q|$
  neighbors in $G$. 
\end{reduction rule}

The proof of the following lemma is obvious and thus omitted.

\begin{lemma}
  Reduction Rule~\ref{red-rule:large-non-neighbor} is safe and can be
  implemented in polynomial time.
\end{lemma}
Note that if Reduction Rule~\ref{red-rule:large-non-neighbor} can no
longer be applied, then every vertex in $D'$ has at least
$|Q|$ neighbors, which because of $|Q|\geq |C|$ implies that every such
vertex has at most $|D|\leq k$ non-neighbors in $G$ and hence also in
$C$. Let $C_N$ be the set of all vertices in $C$ that are not adjacent
to all vertices in $D'$ and let $C'=C-C_N$. Note that
$|C_N|\leq |D||D|\leq k^2$.

We show next how to reduce the size of $C_N$ to $k$. Note that this step is
optional if our aim is solely to obtain a polynomial kernel, however, it
allows us to reduce the number of vertices in the resulting kernel
from $\bigoh(k^2)$ to $\bigoh(k)$. Let $J$ be the bipartite graph
with partition $(C_N,D)$ having an edge between $c \in C_N$ and $d \in
D$ if $\{c,d\} \notin E(G)$.
\begin{reduction rule}\label{rr:precol-crown}
  If $A \subseteq C_N$ is an inclusion-wise minimal set satisfying $|A|>|N_J(A)|$, then
  remove the vertices in $D' \cap N_J(A)$ from $G$. 
\end{reduction rule}
Note that after the application of  Reduction
Rule~\ref{rr:precol-crown}, the vertices in $A$ are implicitly removed
from $C_N$ and added to $C'$ since all their non-neighbors in $D'$
(i.e. the vertices in $D'\cap N_J(A)$) are removed from the graph.
\begin{lemma} \label{lemma:precol-crown}
  Reduction Rule~\ref{rr:precol-crown} is safe and can be
  implemented in polynomial time.
\end{lemma}
\begin{proof}
  It is clear that the rule can be implemented in polynomial-time.
  Towards showing the safeness of the rule, it suffices to show that
  $G$ has a coloring extending $\lambda_P$ using only colors from $Q$
  if and only if so does $G \setminus (D'\cap N_J(A))$. Since $G\setminus
  (D' \cap N_J(A))$ is
  a subgraph of $G$, the forward direction of this statement is
  trivial. So assume that $G\setminus (D'\cap N_J(A))$ has a coloring $\lambda$
  extending $\lambda_P$ using only colors from $Q$. Because the set
  $A$ is inclusion-minimal, we obtain from
  Proposition~\ref{pro:hallset}, that there is a (maximum) matching,
  say $M$, 
  between $N_J(A)$ and $A$ in $J$ that saturates $N_J(A)$. Moreover,
  it follows from the definition of $J$ that every vertex in $A$ is
  adjacent to every vertex in $D\setminus N_J(A)$ in the graph $G$.
  Hence, we obtain that every color in $\lambda(A)$ appears
  exactly once. Hence, we can extend $\lambda$ into a coloring
  $\lambda'$ for $G$ by coloring the vertices in $D'\cap N_J(A)$ according to
  the matching $M$. More formally, let $\lambda_{D'\cap N_J(A)}$ be the
  coloring for the vertices in $D' \cap N_J(A)$ by setting
  $\lambda_{D'\cap N_J(A)}(v)=\lambda(u)$ for every $v \in D'\cap N_J(A)$, where $\{v,u\} \in M$.
  Then, we obtain $\lambda'$ by setting: $\lambda'(v)=\lambda(v)$ for
  every $v \in V(G)\setminus (D'\cap N_J(A))$ and
  $\lambda'(v)=\lambda_{D' \cap N_J(A)}(v)$ for every
  vertex $v \in D' \cap N_J(A)$.
\end{proof}
Note that because of Proposition~\ref{pro:hallset},
we obtain that there is a set $A \subseteq C_N$ with $|A|>|N_J(A)|$ as
long as $|C_N|>|D|$. Moreover, since $N_J(A)\cap D'\neq \emptyset$ for
every such set $A$ (due to the definition of $C_N$), we obtain that
Reduction
Rule~\ref{rr:precol-crown} is applicable as long as
$|C_N|>|D|$. Hence after an exhaustive application of Reduction
Rule~\ref{rr:precol-crown}, we obtain that $|C_N|\leq |D'|\leq k$.

We now introduce our final two reduction rules, which allow us to reduce
the the size of $C'$.
\begin{reduction rule}\label{red-rule:precolCP}
  Let $v \in V(C')$ be a precolored vertex with color $\lambda_P(v)$. Then remove
  $\lambda_P^{-1}(\lambda_P(v))$ from $G$ and $\lambda_P(v)$ from $Q$.
\end{reduction rule}

\begin{lemma}
  Reduction Rule~\ref{red-rule:precolCP} is safe and can be
  implemented in polynomial time.
\end{lemma}
\begin{proof}
  Because $v \in V(C')$, it holds that only vertices in $D_P$ can have
  color $\lambda_P(v)$, but these are already precolored. Hence in any
  coloring for $G$ that extends $\lambda_P$, the vertices in
  $\lambda_P^{-1}(\lambda_P(v))$ are the only vertices that obtain color $\lambda_P(v)$,
  which implies the safeness of the rule.
\end{proof}
Because of Reduction Rule~\ref{red-rule:precolCP}, we can from now on
assume that no vertex in $C'$ is precolored.

Note that the only part of $G$, whose size is not yet bounded by a
polynomial in the parameter $k$ is $C'$. To reduce the size of $C'$,
we need will make use of Proposition~\ref{pro:matching}.

Let $P=\lambda_P(D_P)$ and $H$ be the bipartite graph with bipartition $(C',P)$ containing an
edge between $c' \in C'$ and $p \in P$ if and only if $c'$ is not
adjacent to a vertex pre-colored by $p$ in $G$.
\begin{reduction rule}
  \label{red-rule:matching}
  Let $M$ be a maximum matching in $H$ and let $C_M$ be the endpoints
  of $M$ in $C'$. Then remove all vertices in $C_{\overline{M}}:=C'\setminus C_M$ from
  $G$ and remove an arbitrary set of $|C_{\overline{M}}|$ colors from
  $Q\setminus \lambda_P(X)$. (Recall that $\lambda_P:\ X\rightarrow Q.$)
\end{reduction rule}
In the following let $C_M$ and $C_{\overline{M}}$ be as defined in the above
reduction rule for an arbitrary maximum matching $M$ of $H$. To show
that the reduction rule is safe, we need the following auxiliary
lemma, which shows that if a coloring for $G$ reuses colors from $P$
in $C'$, then those colors can be reused solely on the vertices in $C_M$.
\begin{lemma}
  \label{lem:matching}
  If there is a coloring $\lambda$ for $G$ extending $\lambda_P$ using only
  colors in $Q$, then there is a coloring
  $\lambda'$ for $G$ extending $\lambda_P$ using only colors in $Q$
  such that $\lambda'(C_{\overline{M}})\cap P=\emptyset$.
\end{lemma}
\begin{proof}
  Let $C_P$ be the set of all vertices $v$ in $C'$ with $\lambda(v)
  \in P$. If $C_P\cap C_{\overline{M}}=\emptyset$, then setting $\lambda'$ equal to
  $\lambda$ satisfies the claim of the lemma. Hence assume that $C_P\cap C_{\overline{M}}\neq
  \emptyset$. Let $N$ be the matching in $H$ containing the edges
  $\{v,\lambda(v)\}$ for every $v \in C_P$; note that $N$ is indeed a
  matching in $H$, because $C_P$ is a clique in $G$. Because of
  Proposition~\ref{pro:matching}, there is a matching $N'$ in
  $H[C_M\cup P]$ such that $N'$ has exactly the same endpoints in $P$
  as $N$. Let $C_{M}[N']$ be the endpoints of $N'$ in $C_M$ and let
  $\lambda_A$ be the coloring of the vertices in $C_{M}[N']$
  corresponding to the matching $N'$, i.e., a vertex $v$ in $C_{M}[N']$
  obtains the unique color $p \in P$ such that $\{v,p\} \in N'$.
  Finally, let $\alpha$ be an arbitrary bijection between the vertices in
  $(V(N)\cap C')\setminus C_M[N']$ and the vertices in $C_M[N']\setminus (V(N)\cap C')$, which
  exists because $|N|=|N'|$.
  We now obtain $\lambda'$ from $\lambda$ by setting
  $\lambda'(v)=\lambda_A(v)$ for every $v \in C_{M}[N']$,
  $\lambda'(v)=\lambda(\alpha(v))$ for every vertex $v \in
  (V(N)\cap C')\setminus C_M[N']$, and $\lambda'(v)=\lambda(v)$ for every
  other vertex. To see that $\lambda'$ is a proper coloring note that
  $\lambda'(C')=\lambda(C')$. Moreover, all the colors in
  $\lambda(C')\setminus P$ are ``universal colors'' in the sense that
  exactly one vertex of $G$ obtains the color and hence those
  colors can be freely moved around in $C'$. Finally, the matching
  $N'$ in $H$ ensures that the vertices in $C_M[N']$ can be colored
  using the colors from $P$.
\end{proof}

\begin{lemma}
  \label{lemma:matching-rule-is-safe}
  Reduction Rule~\ref{red-rule:matching} is safe and can be implemented in polynomial time.
\end{lemma}
\begin{proof}
  Note first that the reduction can always be applied
  since if $Q\setminus \lambda_P(X)$ contains less than $|C_{\overline{M}}|$
  colors, then the instance is a no-instance.
  It is clear that the rule can be implemented in polynomial time
  using any polytime algorithm for finding a maximum
  matching. 
  Moreover, if the reduced graph has a coloring extending
  $\lambda_P$ using only the colors in $Q$, then so does the original
  graph, since the vertices in $C_{\overline{M}}$ can be colored with the colors
  removed from the original instance.

  Hence, it remains to show that if $G$ has a coloring, say $\lambda$,
  extending $\lambda_P$ using only colors in $Q$, then $G \setminus
  C_{\overline{M}}$ has a coloring extending $\lambda_P$ that uses only colors in
  $Q':=Q\setminus Q_{\overline{M}}$, where $Q_{\overline{M}}$ is the set of $|C_{\overline{M}}|$ colors from $Q
  \setminus \lambda_P(X)$ that have been removed from $Q$.

  Because of Lemma~\ref{lem:matching}, we may assume that
  $\lambda(C_{\overline{M}}) \cap P=\emptyset$. Let $B$ be the set of all vertices $v$
  in $G-C_{\overline{M}}$ with $\lambda(v) \in Q_{\overline{M}}$. If $B=\emptyset$, then $\lambda$ is a
  coloring extending $\lambda_P$ using only colors from $Q'$. Hence assume that
  $B\neq \emptyset$. Let $A$ be the set of all vertices $v$ in $C_{\overline{M}}$ with
  $\lambda(v) \in Q'$.
  Then $\lambda(A) \cap \lambda_P(X)=\emptyset$, which implies
  that every color in $\lambda(A)$ appears only in $C_{\overline{M}}$ (and exactly once in
  $C_{\overline{M}}$). Moreover, $|\lambda(A)|\geq |\lambda(B)|$. Let $\alpha$ be an arbitrary
  bijection between $\lambda(B)$ and an arbitrary subset of $\lambda(A)$ (of size $|B|$)
  and let
  $\lambda'$ be the coloring obtained from $\lambda$ by setting
  $\lambda'(v)=\alpha(\lambda(v))$ for every $v \in B$,
  $\lambda'(v)=\alpha^{-1}(\lambda(v))$ for every $v \in A$, and
  $\lambda'(v)=\lambda(v)$, otherwise. Then $\lambda'$ restricted to
  $G-C_{\overline{M}}$ is a coloring for $G-C_{\overline{M}}$ extending $\lambda_P$ using only
  colors from $Q'$. Note that $\lambda'$ is a proper coloring because
  the colors in $\lambda(A)$ are not in $P$ and hence do not appear
  anywhere else in $G$ and moreover the colors in $\lambda(B)$ do
  not appear in $\lambda(C_{\overline{M}})$.
\end{proof}
Note that after the application of Reduction Rule~\ref{red-rule:matching},
it holds that $|C'|=|C_{M}|\leq |P|\leq |D_P|\leq |D|\leq k$. Together with the
facts that $|D|\leq k$, $|C_N|\leq k$, we obtain that the reduced
graph has at most $3k$ vertices.
\begin{theorem}\label{the:pcecm-kernel}
  \PCECM{} admits a polynomial kernel with at most $3k$ vertices.
\end{theorem}

\section{Polynomial kernel and Compression for {\RLC}}\label{sec:regpk}

We now show our polynomial kernel and compression for \RLC{}, which is
more intricate than the one for \PCECM{}.  We begin by noting, via
Banik et al.~\cite{IWOCApaper}, that we can compute a clique modulator
of size at most $2k$. Let $(G,k,L)$ be an input of \RLC.

\begin{reduction rule}
  \label{rr:rlc-cliquemod}
  If $\overline{G}$ has a matching of $k$ edges, then return a dummy
  constant-sized positive instance.  Otherwise let $M$ be a maximal
  matching in $\overline{G}$ and use $V(M)$ as a clique modulator.
\end{reduction rule}

\begin{lemma}[\cite{IWOCApaper}]
  Reduction Rule~\ref{rr:rlc-cliquemod} is correct and can be
  implemented in polynomial time.
\end{lemma}

Henceforth, we let $V(G)=C \cup D$ where $G[C]$ is a clique and $D$ is
a clique modulator, $|D| \leq 2k$. Let $T=\bigcup_{v \in V(G)}L(v)$.
We note one further known reduction rules for \RLC{}.
Consider the bipartite graph $H_G$ with bipartition $(V(G),T)$ having an
edge between $v \in V(G)$ and $t \in T$ if and only if $t \in L(v)$.
\begin{reduction rule}[{\cite{IWOCApaper}}]
  \label{rr:lc-colors}
  Let $T'$ be an inclusion-wise minimal subset of $T$ such that
  $|N_{H_G}(T')|<|T'|$, then remove all vertices in $N_{H_G}(T')$ from $G$.
\end{reduction rule}
Note that after an exhaustive application of Reduction
Rule~\ref{rr:lc-colors}, it holds that $|T|\leq |V(G)|$ since
otherwise Proposition~\ref{pro:hallset} would ensure the applicability
of the reduction rule. Hence in the following we will assume that
$|T|\leq |V(G)|$.

With this preamble handled, let us proceed with the kernelization.
We are not able to produce a direct `crown reduction rule' for 
\textsc{List Coloring}, as for \textsc{Pre-Coloring Extension}
(e.g., we do not know of a useful generalization of Reduction
Rule~\ref{rr:precol-crown}).
Instead, we need to study more closely which list colorings
of $G[D]$ extend to list colorings of $G$. 
For this purpose, let $H=H_G-D$ be the bipartite graph with bipartition $(C,T)$ having an edge
$\{c,t\}$ with $c \in C$ and $t \in T$ if and only if $t \in L(c)$.
Say  that a partial list coloring $\lambda_0 \colon A \to T$ is \emph{extensible}
if it can be extended to a proper list coloring $\lambda$ of $G$.
If $D \subseteq A$, then a sufficient condition for this is that $H-(A \cup \lambda_0(A))$ 
admits a matching saturating $C \setminus A$. 
(This is not a necessary condition, since some colors used in $\lambda_0(D)$
could be reused in $\lambda(C \setminus A)$, but this investigation will point in the
right direction.)
By Proposition~\ref{pro:hallset}, this is characterized by Hall sets in $H-(A \cup \lambda_0(A))$.

A Hall set $S \subseteq U$ in a bipartite graph $G'$ with bipartition
$(U,W)$ is \emph{trivial} if $N(S)=W$. 
We start by noting that if a color occurs in sufficiently many
vertex lists in $H$, then it behaves uniformly with respect to 
extensible partial colorings $\lambda_0$ as above.
\begin{lemma}
  \label{lemma:unrare-is-safe}
  Let $\lambda_0 \colon A \to T$ be a partial list coloring where 
  $|A \cap C| \leq p$ and let $t \in T$ be a color that occurs in at
  least $k+p$ lists in $C$.
  Then $t$ is not contained in any non-trivial Hall set of colors 
  in $H-(A \cup \lambda_0(A))$.
\end{lemma}
\begin{proof}
  Let $H'=H-(A \cup \lambda_0(A))$. 
  Consider any Hall set of colors $S \subset (T \setminus \lambda_0(A))$
  and any vertex $v \in C \setminus (A \cup N_{H'}(S))$ (which exists
  assuming $S$ is non-trivial). Then 
  $S \subseteq T \setminus L(v)$, hence $|S| \leq k$,
  and by assumption $|N_{H'}(S)|<|S|$. But for every $t' \in S$,
  we have $N_H(t') \subseteq N_{H'}(S) \cup (A \cap C)$, hence
  $t'$ occurs in at most $|N_{H'}(S) \cup (A \cap C)| < k+p$ vertex 
  lists in $C$. Thus $t \notin S$. 
\end{proof}

In the following, we will assume that $n\geq 11k$.\footnote{The constants $11k$ and $6k$
  in this paragraph are chosen to make the arguments work smoothly. A smaller kernel is possible with a more
  careful analysis and further reduction rules.}
This is safe, since otherwise (by Reduction Rule~\ref{rr:lc-colors}) 
we already have a kernel with a linear number of vertices and colors.   
We say that a color $t \in T$ is \emph{rare} if it occurs in at most
$6k$ lists of vertices in $C$.

\begin{lemma}
  If $n \geq 11k$, then there are at most $3k$ rare colors.
\end{lemma}
\begin{proof}
  Let $S=\{t \in T \mid d_H(t)<6k\}$. For every $t \in S$,
  there are $|C|-6k$ ``non-occurrences'' (i.e., vertices
  $v \in C$ with $t \notin L(v)$), and there are $|C|k$ non-occurrences in
  total. Thus 
  \[
  |S| \cdot (|C|-6k) \leq |C|k \quad \Rightarrow \quad |S| \leq \frac{|C|}{|C|-6k}k
  = (1+\frac{6k}{|C|-6k})k,
  \]
  where the bound is monotonically decreasing in $|C|$ and maximized (under
  the assumption that $n\geq 11k$ and hence $|C|\geq 9k$) for $|C|=9k$ yielding $|S| \leq 3k$.
\end{proof}

Let $T_R \subseteq T$ be the set of rare colors. 
Define a new auxiliary bipartite graph $H^*$ with bipartition $(C,D
\cup T_R)$ having an edge between a vertex $c \in C$ and a vertex $d
\in D$ if $\{c,d\} \notin E(G)$ and an edge between a vertex $c \in C$
and a vertex $t \in T_R$ if $t \in L(c)$.
Refer to the colors $T_R \setminus X$ as \emph{constrained} rare colors. 
Note that constrained rare colors only occur on lists of vertices in
$D \cup (C \cap X)$. 

Let $T'=T\setminus (T_R\setminus X)$, $V'=(D\setminus
X)\cup(C\cap X)$, and set $q=|T'|-|C \setminus X|$.
Before we continue, we want to provide some useful
observations about the sizes of the considered sets and numbers.
\begin{observation}\label{obs:sizes}
  \begin{itemize}
  \item $|X|\leq |D|+|T_R|\leq 5k$,
  \item $|V'|\leq |D|+|X|\leq 7k$,
  \item $q \leq |T| - |C| + |C \cap X| \leq 
    |D|+|X|\leq 7k$; this holds because $|T| \leq |V|=|C|+|D|$.
  \end{itemize}
\end{observation}

\begin{lemma}\label{lem:compression}
  Assume $n \geq 11k$. Then
  $G$ has a list coloring if and only if there is a partial 
  list coloring $\lambda_0 \colon V' \to T$
  that uses at most $q=|T'|-|C\setminus X|$ colors from $T'$.
\end{lemma}
\begin{proof}
  The number of colors usable in $C \setminus X$
  is $|T'|-p$ where $p$ is the number counted above
  (since constrained rare colors cannot be used in $C \setminus X$
  even if they are unused in $\lambda_0$). 
  Thus it is a requirement that $|T'|-p \geq |C \setminus X|$.
  That is, $p \leq |T'|-|C \setminus X|=q$. 
  Thus necessity is clear. 
  We show sufficiency as well. That is, let $\lambda_0$ be a partial 
  list coloring with scope $V'=(C \cap X) \cup (D \setminus X)$
  which uses at most $q$ colors of $T'$.
  We modify and extend $\lambda_0$ to a list coloring of $G$.

  First let $H_0$ be the bipartite graph with bipartition
  $(V,T_R\setminus X)$ and let $M_0$ be a matching saturating
  $T_R\setminus X$; note that this
  exists by reduction rule~\ref{rr:lc-colors}. We modify $\lambda_0$ 
  to a coloring $\lambda_0'$ so that every constrained rare color is used
  by $\lambda_0'$, by iterating over every color $t \in T_R \setminus X$; for every $t$,
  if $t$ is not yet used by $\lambda_0'$, then let $vt \in M_0$ and 
  update $\lambda_0'$ with $\lambda_0'(v)=t$. Note that the scope
  of $\lambda_0'$ after this modification is contained in 
  $(C \cap X) \cup D$.  Next, let $M$ be a maximum matching in $H^*$. 
  We use $M$ to further extend $\lambda_0'$ in stages to a partial list
  coloring $\lambda$ which colors all of $D$ and uses all rare colors. 
  In phase 1, for every color $t \in T_R \cap X$ which is not already
  used, let $vt \in M$ be the edge covering $t$ and assign $\lambda(v)=t$.
  Note that $M$ matches every vertex of $X$ in $H^*$ with a vertex not in $X$, 
  thus the edge $vt$ exists and $v$ has not yet been assigned in $\lambda$.
  Hence, at every step we maintain a partial list coloring,
  and at the end of the phase all rare colors have been assigned. 
  Finally, as phase 2, for every vertex $v \in D \cap X$ not yet assigned,
  let $uv \in M$ where $u \in C$; necessarily $u \in C \setminus X$ and $u$ is 
  as of yet unassigned in $\lambda$. The number of colors assigned in
  $\lambda$ thus far is at most $|X|+|D| \leq |T_R|+2|D|\leq 7k$, whereas
  $|L(u) \cap L(v)| \geq n-2k \geq 9k$, hence there always exists 
  an unused shared color that can be mapped to $\lambda(u)=\lambda(v)$.
  Let $\lambda$ be the resulting partial list coloring.
  We claim that $\lambda$ can be extended to a list coloring of $G$. 
  
  Let $A$ be the scope of $\lambda$ and let $H'=H-(A \cap \lambda(A))$.
  Note that $A \cap C \subseteq V(M)$, hence $|A \cap C|\leq |D|+|T_R|\leq 5k$.
  Thus by Lemma~\ref{lemma:unrare-is-safe}, no non-trivial Hall set in $H'$ can contain a rare color.
  However, all rare colors are already used in $\lambda$. Thus $H'$
  contains no non-trivial Hall set of colors. Thus the only possibility
  that $\lambda$ is not extensible is that $H'$
  has a trivial Hall set, i.e., $|T \setminus \lambda(A)|<|C \setminus A|$.
  However, note that every modification after $\lambda_0'$ added one vertex
  to $A$ and one color to $\lambda(A)$, hence the balance between the two
  sides is unchanged. Thus already the partial coloring $\lambda_0'$ leaves behind
  a trivial Hall set. However, $\lambda_0'$ colors precisely $C \cap X$ in $C$
  and leaves at least $|T'|-q$ colors remaining. 
  By design this is at least $|C \setminus X|$, yielding a contradiction. 
  Thus we find that $H'$ contains no Hall set, and $\lambda$
  is a list coloring of $G$. 
\end{proof}

Before we give our compression and kernelization results, we need the
following simple auxiliary lemma.
\begin{lemma}\label{lem:tprime-universal}
  $T'$ contains at least $|T'|-|V'|k$ colors that
  are universal to all vertices in $V'$.
\end{lemma}
\begin{proof}
  The list of every vertex $v \in V'$ misses at most $k$ colors from
  $T'$. Hence all but at most $|V'|k$ colors in $T'$ are universal to
  all vertices in $V'$.
\end{proof}

For clarity, let us define the output problem of our compression explicitly.

\begin{problem}[]{\textsc{Budget-Constrained List Coloring}}
  \Input & A graph $G$, a set $T$ of colors, a list $L(v) \subseteq T$ for every $v \in V(G)$,
  and a pair ($T', q)$ where $T' \subseteq T$ and $q \in \mathbb{N}$.\\
  \Prob  & Is there a proper list coloring for $G$ that uses at most $q$ distinct colors
  from $T'$?
\end{problem}

\begin{theorem}
  \RLC{} admits a compression into an instance of \textsc{Budget-Constrained List Coloring}
  with at most $11k$ vertices and $\bigoh(k^2)$ colors, encodable in $\bigoh(k^2 \log k)$ bits. 
\end{theorem}
\begin{proof}
  If $|V(G)|\leq 11k$, then $G$ itself can be used as the output (with a dummy budget
  constraint). Otherwise, all the bounds above apply and 
  Lemma~\ref{lem:compression} shows that the existence of a list
  coloring in $G$ is equivalent to the existence of a list coloring in
  $G[V']$ that uses at most $q$ colors from $T'$. Since $|V'|\leq 7k$,
  it only remains to reduce the number of colors
  in $T_R\cup T'$. Clearly, if $|T'|<|V'|k+q$, then $|T_R\cup T'|\leq
  3k+(7k)k\in \bigoh(k^2)$ and there is nothing left to show.
  So suppose that $|T'|\geq|V'|k+q$. Then, it follows from
  Lemma~\ref{lem:tprime-universal} that $T'$ contains at
  least $q$ colors that are universal to the vertices in $V'$ and we
  obtain an equivalent instance by removing all but exactly $q$
  universal colors from $T'$, which leaves us with an instance with at
  most $|T_R|+q\leq 3k+7k^2\in \bigoh(k^2)$ colors, as required. 
  Finally, to describe the output concisely, note that $G[V']$ can be trivially described
  in $\bigoh(k^2)$ bits, and the lists $L(v)$ can be described by enumerating
  $T \setminus L(v)$ for every vertex $v$, which is $k$ colors per vertex,
  each color identifiable by $\bigoh(\log k)$ bits. 
\end{proof}

Note that the compression is asymptotically essentially optimal, since even the basic
\textsc{4-Coloring} problem does not allow a compression in $\bigoh(n^{2-\varepsilon})$
bits for any $\varepsilon > 0$ unless the polynomial hierarchy collapses~\cite{JansenP17}.
For completeness, we also give a proper kernel.

\begin{theorem}
  \RLC{} admits a kernel with $\bigoh(k^2)$ vertices and colors.
\end{theorem}
\begin{proof}
  We distinguish two cases depending on whether or not $|T'|< |V'|k+q$.
  If $|T'|< |V'|k+q$, then $|T|\leq
  |T_R|+|T'|< 3k+|V'|k+q\leq 3k+(7k)(k+1)\in \bigoh(k^2)$.
  Since a list coloring requires at least one distinct color for
  every vertex in $C$, it holds that $|C|\leq |T|\leq  3k+(7k)(k+1)$ and hence
  $|V(G)|\leq (3+7k)k+2k \in \bigoh(k^2)$, implying the desired kernel.

  If on the other hand, $|T'|\geq |V'|k+q$, then, because of
  Lemma~\ref{lem:tprime-universal} it holds that $T'$ contains a set
  $U$ of exactly $q$ colors that are universal to the vertices in $V'$.
  Recall that Lemma~\ref{lem:compression} shows that the existence of a list
  coloring in $G$ is equivalent to the existence of a list coloring in
  $G[V']$ that uses at most $q=|T'|-|C \setminus X|$ colors from $T'$.
  It follows that the graph $G[V']$ has a list coloring using only
  colors in $(T_R\setminus X)\cup U$ if and only if $G$ has a list
  coloring.
  Hence, it only remains to restore the regularity of the instance.
  We achieve this as follows.
  First we add a set $T_N$ of $|(T_R\setminus X)\cup U|$ novel colors. We then
  add these colors (arbitrarily) to the color lists of the vertices in
  $V'$ such that the size of every list (for any vertex in $V'$) is
  $|(T_R\setminus X)\cup U|$. This clearly already makes the instance regular,
  however, now we also need to ensure that no vertex in $V'$ can be
  colored with any of the new colors in $T_N$. To achieve this we add a set $C_N$ of $|T_N|$
  novel vertices to $G[V']$, which we connect to every vertex in
  $(C\cap X)\cup C_N$ and whose lists all contain all the new colors
  in $T_N$. It is clear that the constructed instance is equivalent to
  the original instance since all the new colors in $T_N$ are required to color
  the new vertices in $C_N$ and hence no new color can be used to
  color a vertex in $V'$. Moreover, $D$ is still a clique modulator
  and the number $k'$ of missing colors (in each list of
  the constructed instance) is equal to $|D|+|C\cap X|\leq
  2k+5k$ because the instance is $(n-|D|-|C\cap X|)$-regular.
  Finally, the instance has at most $|V'\cup
  C_N|\leq 7k+3k+7k=17k \in \bigoh(k)$ vertices and at most
  $2(|T_R|+|U|)\leq 2(3k+7k)=20k\in \bigoh(k)$ colors, as required.
\end{proof}

\section{Saving $k$ colors: Pre-coloring and List Coloring Variants}
\label{sec:save-col}

In this section, we consider natural pre-coloring and list coloring
variants of the ``saving $k$ colors'' problem, defined as:

\begin{problem}[$k$]{\SKC}
  \Input & A graph $G$ with $n$ vertices and an integer $k$.\\
  
  \Prob  & Does $G$ have a proper coloring using at most $n-k$ colors?
\end{problem}

This problem is known to be {\sf FPT}  (it even allows
for a linear kernel)~\cite{chorFJ04}, when parameterized by $k$. Notably the problem
provided the main motivation for the introduction of \RLC{}
in~\cite{IWOCApaper,AroraBP018}.

We consider the following (pre-coloring and list coloring) extensions of \SKC{}.

\begin{problem}[$n-|Q|$]{\PCESN}
  \Input & A graph $G$ with $n$ vertices and a pre-coloring
  $\lambda_P: X \rightarrow Q$ for $X \subseteq V(G)$ where $Q$ is a
  set of colors.\\
  
  \Prob  & Can $\lambda_P$ be extended to a proper coloring of $G$
  using only colors from $Q$?
\end{problem}

\begin{problem}[$k$]{\LCNminusK}
  \Input & A graph $G$ on $n$ vertices with a list $L(v)$ of colors for every $v \in V(G)$ and an integer $k.$\\
  \Prob  & Is there a proper list coloring of $G$ using at most $n-k$ colors?
\end{problem}

Note that the following variant seems natural, however, is trivially
NP-complete even when the parameter $k$ is equal to $0$, since the
problem with an empty precoloring then corresponds to the problem
whether $G$ can be colored by at most $|Q|$ colors.

\begin{problem}[$k$]{\PCESQ}
  \Input & A graph $G$ with $n$ vertices, a pre-coloring
  $\lambda_P: X \rightarrow Q$ for $X \subseteq V(G)$ where $Q$ is a
  set of colors, and an integer $k$.\\
  
  \Prob  & Can $\lambda_P$ be extended to a proper coloring of $G$
  using at most $|Q|-k$ colors from $Q$?
\end{problem}

Interestingly, we will show that \PCESN{} is {\sf FPT}
and even allows a linear kernel. Thus, we generalize the above-mentioned result of Chor et al.~\cite{chorFJ04}
(set $Q=[n-k]$ and $X=\emptyset$). However,  \LCNminusK{} is W[1]-hard.

\begin{theorem}
  {\PCESN{}} (parameterized by $n-|Q|$) has a kernel with at most $6(n-|Q|)$
  vertices and is hence fixed-parameter tractable.
\end{theorem}
\begin{proof}
  Let $G'$ be the graph obtained from $G$ after applying the following
  preprocessing rules:
  \begin{description}
  \item[(R1)] If $u$ and $v$ are two distinct vertices in $G\setminus X$
    such that $\lambda_P(N_G(u)) \cup \lambda_P(N_G(v))=Q$, then we add
    an edge between $u$ and $v$ in $G$. This rule is safe because $u$ and
    $v$ cannot be colored with the same color.
    \item[(R2)] If $u$ is a vertex in $G\setminus X$ that is adjacent to a
    vertex $v \in X$, then we can safely add all edges between $u$ and every vertex
    in $\lambda_P^{-1}(\lambda_P(v))$. 
  \item[(R3)] If $u$ and $v$ are two distinct vertices in $X$ such that
    $\lambda_P(u)\neq \lambda_P(v)$, then we can again safely add an
    edge between $u$ and $v$.
   \end{description}
  
  Let $M$ be a maximal matching in the complement of
  $G'$. Note that if $|M|\leq n-|Q|$, then $V(M)$ is a clique modulator
  for $G'$ of size at most $2(n-|Q|)$ and we obtain a kernel with at most $6(n-|Q|)$ vertices using
  Theorem~\ref{the:pcecm-kernel}.
  Thus assume that $|M|\geq n-|Q|$. In this case we can safe $|M|\geq
  n-|Q|$ colors by giving the endpoints of every edge in $M$ the same
  color. Namely, let $\{u,v\} \in M$, then:
  \begin{itemize}
  \item if $u,v \notin X$, then it follows from (R1) that there is a
    color $q \in Q$ that can be given to both vertices,
  \item if $u \notin X$ and $v \in X$, then it follows from (R2) that
    we can color $u$ with color $\lambda_P(v)$,
  \item if $u,v \in X$, then by (R3) we have that $\lambda_P(u)=\lambda_P(v)$.
  \end{itemize}
  Note that after coloring the edges in $M$ with the same color,
  removing $V(M)$ from $G'$, and removing the colors used for the
  edges in $M$ from $Q$, the number of colors in the remaining
  instance is equal to the number of vertices in the remaining
  instance, implying that the remaining instance can be properly colored.
\end{proof}

Finally, we show W[1]-hardness of \LCNminusK{} using a parameterized
reduction from \IS{}.
\begin{theorem}\label{the:list-save-hard}
  {\LCNminusK} parameterized by $k$ is W[1]-hard.
\end{theorem}
\begin{proof}
  We provide a parameterized reduction from \IS{}, which is a
  well-known W[1]-complete problem~\cite{cygan2015}.

  \begin{problem}[$k$]{\IS}
    \Input & A graph $G$ and an integer $k$.\\
    \Prob  & Does $G$ have an independent set of size at least $k$?
  \end{problem}

  Let $(G,k)$ be the given instance of \IS{} and suppose that
  $V(G)=\{v_1,\dotsc,v_n\}$.
  We claim that the instance $(G,L,k')$ of
  \LCNminusK{} with $k'=k+1$ and $L(v_i)=\{i,n+1\}$  is equivalent
  to $(G,k)$. Note that $L$ assigns each
  vertex a private color $i$ as well as a shared color $n+1$.

  Suppose that $G$ has an independent set $X$ of size $k$.
  Let $\lambda \colon V(G) \rightarrow [n+1]$ be the coloring defined
  by setting $\lambda(v_i)=n+1$ if $v_i \in X$ and $\lambda(v_i)=i$,
  otherwise. Then because $X$ is an independent set, we obtain that
  $\lambda$ is a proper list coloring for $G$ that uses
  $n+1-k=n-(k+1)=n-k'$ colors, as required.

  Towards showing the reverse direction, let $\lambda$ be a proper
  list coloring for $G$ that uses at most $n-k'$ colors. Because $n+1$
  is the only color occuring in the list of more than one vertex, it
  follows that at least $k'-1=k$ vertices must be colored with color
  $n+1$. Moreover, because $\lambda$ is a proper coloring, the set $X$
  of vertices that are colored with $n+1$ by $\lambda$ forms an
  independent set of $G$, having size is at least $k$, as required.
\end{proof}

\section{Conclusions} \label{sec:conc}

We have showed several results regarding the parameterized complexity
of \textsc{List Coloring} and \textsc{Pre-Coloring Extension} problems. 
We showed that \textsc{List Coloring}, and hence also
\textsc{Pre-Coloring Extension}, parameterized by the size of a clique
modulator admits a randomized FPT algorithm
with a running time of $\bigoh^*(2^k)$, matching the best known
running time of the basic \textsc{Chromatic Number} problem
parameterized by the number of vertices.
This answers open questions of Golovach et al.~\cite{golovachPS14}.
Additionally, we showed that \textsc{Pre-Coloring Extension}
under the same parameter admits a linear vertex kernel with at most $3k$
vertices and that \RLC{} admits a compression into a problem 
we call \textsc{Budget-Constrained List Coloring}, into an instance
with at most $11k$ vertices, encodable in $\bigoh(k^2 \log k)$ bits. 
The latter also admits a proper kernel with $\bigoh(k^2)$ vertices and
colors. This answers an open problem of Banik et al.~\cite{IWOCApaper}.

One open question is to optimize the bound $11k$ on the number of
vertices in the \RLC{} compression, and/or show a proper kernel with
$\bigoh(k)$ vertices. Another set of questions is raised by
Escoffier~\cite{Escoffier19}, who studied the \textsc{Max Coloring}
problem from a ``saving colors'' perspective. In addition to the
questions explicitly raised by Escoffier, it is natural to ask whether
his problems \textsc{Saving Weight} and \textsc{Saving Color Weights}
admit FPT algorithms with a running time of $2^{\bigoh(k)}$ and/or
polynomial kernels.

\end{document}